\documentclass[11pt]{article}
\usepackage{anysize}
\usepackage{latexsym}
\usepackage{amsmath}
\usepackage{amssymb}
\usepackage{amsthm}
\usepackage{epsfig}
\usepackage{algorithm}
\usepackage{algorithmic}
\marginsize{1in}{1in}{1in}{1in}

\newtheorem{theorem}{Theorem}

\newtheorem{lemma}[theorem]{Lemma}

\begin{document}
\title{Complexity of Unconstrained $L_2$-$L_p$  Minimization}
\author{Xiaojun Chen\thanks{Department of Applied Mathematics, The Hong Kong Polytechnic
University, Hong Kong, China. E-mail:{\tt
maxjchen@polyu.edu.hk}}\and Dongdong Ge\thanks{Antai School of
Economics and Management, Shanghai Jiao Tong University, Shanghai,
China. E-mail: {\tt ddge@sjtu.edu.cn}} \and Zizhuo
Wang\thanks{Department of Management Science and Engineering,
Stanford University, Stanford, CA 94305. E--mail: {\tt
zzwang@stanford.edu}}\and Yinyu Ye\thanks{Department of Management
Science and Engineering, Stanford University, Stanford, CA 94305;
and Visiting Professor of Department of Applied Mathematics, The
Hong Kong Polytechnic University, Hong Kong, China. E--mail: {\tt
yinyu-ye@stanford.edu}}}
\maketitle
\begin{abstract}
We consider the unconstrained $L_2$-$L_p$ minimization: find a
minimizer of $\|Ax-b\|^2_2+\lambda \|x\|^p_p$ for given $A \in
R^{m\times n}$, $b\in R^m$ and parameters $\lambda>0$, $p\in [0,
1)$. This problem has been studied extensively in variable
selection and sparse least squares fitting for high dimensional
data. Theoretical results show that the minimizers of the
$L_2$-$L_p$ problem have various attractive features due to the
concavity and non-Lipschitzian property of the regularization
function $\|\cdot\|^p_p$. In this paper, we show that the
$L_q$-$L_p$ minimization problem is strongly NP-hard for any $p\in
[0,1)$ and $q\ge 1$, including its smoothed version. On the other
hand, we show that, by choosing parameters $(p,\lambda)$
carefully, a minimizer, global or local, will have certain desired
sparsity. We believe that these results provide new theoretical
insights to the studies and applications of the concave
regularized optimization problems.
\end{abstract}


\hspace{0.8in}  sparse solution reconstruction, bridge estimator.


\section{Introduction}

In this paper, we consider the following $L_2$-$L_p$ minimization
problem:
\begin{equation}\label{pp}
\begin{array}{cc}
\mbox{Minimize}_x  & f_p(x):=\|Ax-b\|_2^2+ \lambda \|x\|_p^p
\end{array}
\end{equation}
where data and parameter $A =(a_1,...,a_n)\in R^{m\times n}, 0\ne b\in
R^m$, $\lambda>0$ and $0\leq p<1$, and variables $x\in R^n$. This
regularized formulation has been studied extensively in variable
selection and sparse least squares fitting for high dimensional
data, see \cite{Char07,Char08,CXY,Fan01,Fou09,Frank93,
Ge,Huang,knight,Lai10,Natarajan95} and references therein. Here,
when $p=0$,
\[\|x\|_0^0=\|x\|_0=|\{i:\ x_i\ne 0\}|\]
that is, the number of nonzero entries in $x$.

The original goal of the model was to find a least squares solution
with fewer nonzero entries for an under-determined linear system
that has more variables than the data measurements. For this
purpose, people considered the regularized $L_2$-$L_0$ problem. For
instance, the variable subset selection method can be viewed as the
$L_2$-$L_0$ problem, which is the most popular method of regression
regularization used in statistics \cite{Frank93}.

However, the $L_0$ regularized problem is difficult to deal with because of the
discrete structure of the $0$-norm, while the solvability of the
$L_2$-$L_p$ problem for $p\in (0,1)$ can be derived from the
continuity and level boundedness of $f_p$. A (global) minimizer of
the $L_2$-$L_p$ problem is also called a bridge estimator in
statistical literature \cite{Frank93} and has various nice
properties including the oracle property \cite{Fan01,Huang,knight}.
Moreover, theoretical results show that in distinguishing zero and
nonzero entries of coefficients in sparse high-dimensional
approximation, the bridge estimators have advantages over the Lasso
estimators that minimize the following convex $L_2$-$L_1$
minimization problem:
\begin{equation}\label{p1}
\begin{array}{cc}
\mbox{Minimize}_x  & f_1(x):=\|Ax-b\|_2^2+ \lambda\|x\|_1.
\end{array}
\end{equation}

Due to these advantages, researchers have been interested in the
$L_p$ regularization problem for $0<p<1$. However, the $L_2$-$L_p$
problem (\ref{pp}) is a nonconvex, non-Lipschitz optimization
problem. There are not many optimization theories on analyzing this type of problems.
Many practical approaches have been developed to tackle the
problem (\ref{pp}), see, e.g., \cite{Char07,Char08,CXY,Huang,Lai10};
but there is no globally convergent algorithm that guarantees to find a
global minimizer or bridge estimator.

To the best of our knowledge, the computational complexity of the
$L_2$-$L_p$ minimization problem remains an open problem. One may
attempt to draw a hardness result from the following problem:
\begin{equation}\label{constpp}
\begin{array}{cc}
\mbox{Minimize}   &  \|x\|^p_p \\
\mbox{Subject to} & Ax=b,
\end{array}
\end{equation}
which is shown in \cite{Ge} to be strongly NP-hard for $p\in [0,1)$;
or the problem
\begin{equation}\label{constppepsilon}
\begin{array}{cc}
\mbox{Minimize}   &  \|x\|_0 \\
\mbox{Subject to} & \|Ax-b\|_2\le \epsilon,
\end{array}
\end{equation}
which is shown in \cite{Natarajan95} to be NP-hard for certain
$\epsilon$. From a complexity theory perspective, an NP-hard
optimization problem with a polynomially bounded objective function
does not admit a polynomial-time algorithm, and a strongly NP-hard
optimization problem with a polynomially bounded objective function
does not even admit a fully-polynomial-time approximation scheme
(FPTAS),  unless P=NP \cite{Vaz}.

Indeed, the $L_2$-$L_p$ problem (\ref{pp}) can be viewed as a
quadratic penalty problem of problem (\ref{constpp}). Intuitively,
solving an unconstrained penalty optimization problem is easier
than solving the constrained optimization problem. Unfortunately,
we show that this is not true. More precisely, we show that
finding a global minimizer of $L_2$-$L_p$ problem (\ref{pp})
remains strongly NP-hard for all $0\le p<1$ and $\lambda>0$,
including its smoothed version. We also extend the strong
NP-hardness result to the $L_q$-$L_p$ minimization problem for
$q\ge 1$.

On the positive side, we present a sufficient condition on the
choice of $\lambda$ for the desired sparsity of all minimizers, global or local, of the
$L_2$-$L_p$ problem for given $(A,b,p)$, as long as their objective value is below that of the
all-zero solution. Under this condition,
any such a local optimal solution of problem (\ref{pp}) is a sparse estimator to the original
problem. This may explain why many methods, e.g., \cite{Char07,Char08,CXY,Huang,Lai10},
have reported encouraging computational results, although what they calculate may not
be a global minimizer.

The remainder of this paper is organized as follows: in Section 2,
we present sufficient conditions on the choice of $\lambda$ to meet
the sparsity requirement of global or local minimizers of the
$L_2$-$L_p$ minimization problem. In general, when $\lambda$ is
sufficiently large with respect to data $(A,b)$ and $p$, the number
of nonzero entries in any minimizer of the problem must be small. In
Section 3, we prove that the $L_q$-$L_p$ minimization problem:
\begin{equation}\label{pp-q}
\begin{array}{cc}
\mbox{Minimize}_x   & f_{q,p}(x):=\|Ax-b\|_q^q+ \lambda\|x\|_p^p
\end{array}
\end{equation}
is strongly NP-hard for any given $0\le p < 1$, $q\ge 1$ and
$\lambda >0$. We then extend our hardness result to its smoothed
version:
\begin{equation}\label{pp-q-e}
\begin{array}{cc}
\mbox{Minimize}_x & f_{q,p,\epsilon}(x):=\|Ax-b\|_q^q+
\lambda\sum^n_{i=1}(|x_i|+\epsilon)^p
\end{array}
\end{equation}
for any given $0< p < 1$, $q\ge 1$, $\lambda>0$ and $\epsilon>0$,
even though the objective function in this case is Lipschitz
continuous. Thus, changing the non-Lipschitz regularization model
(\ref{pp-q}) to a Lipschitz continuous model (\ref{pp-q-e}) gains no
advantage in terms of computational complexity. Finally, we show
that our results are consistent with the existing findings from
statistical literature, but give more specific bounds on choosing
regularization parameters. We also illustrate that for the purpose
of finding a least squares solution with a targeted number of
nonzero entries, finding a local minimizer of problem (\ref{pp}) is
likely to accomplish the same objective as finding a global
minimizer does.

In the rest of the paper, we define $z^0 = 0$ if $z =0$ and $z^0 =
1$ if $z\neq 0$. We use $(x\cdot y)$ to represent the vector
$(x_1y_1, \ldots,x_ny_n)^T\in R^n$ and $\|\cdot\|$ to denote the
$L_2$ norm.

\section{Choosing the parameter $\lambda$ for sparsity}
In applications like variable selection and sparse solution
reconstruction, one wants to find least square estimators with no more than $k$
nonzero entries. On the other hand, one obviously wants to avoid the all-zero solution.
The $L_2$-$L_p$ regularized approach is to first solve $L_2$-$L_p$ problem (\ref{pp})
to find a minimizer. Then, eliminate all variables who have zero values in the minimizer, and solve
the least square problem using only remaining variables.
Thus, the key is to control the support size of minimizers of problem (\ref{pp}) such that it does
not exceed $k$, and this is typically accomplished by selecting a suitable $\lambda$.
We now give a sufficient condition on $\lambda$ for the minimizers of the $L_2$-$L_p$ problem to have
desirable sparsity.

\begin{theorem}\label{Th3.1} Let
\begin{equation}\label{alpha-beta}
\beta(k)
=k^{p/2-1}\left(\frac{2\alpha}{p(1-p)}\right)^{p/2}\|b\|^{2-p},
\quad \alpha =\max_{1\le i\le n} \|a_i\|^2, \quad 1\le k\le n.
\end{equation}
The following statements hold.
\begin{description}
\item{(1)} If $\lambda \ge \beta(k)$, any minimizer $x^*$ of
$L_2$-$L_p$ problem (\ref{pp}) satisfies $\|x^*\|_0< k$ for $k\ge
2$. \item{(2)} If $\lambda \ge \beta(1)$, $x^*=0$ is the unique
minimizer of $L_2$-$L_p$ problem (\ref{pp}). \item{(3)} Suppose
that set $C: =\{ \, x \, | \, Ax=b\, \}$ is non-empty. Then, if
$\lambda \le \frac{\|b\|^2}{\|x_c\|_p^p}$  for some $x_c\in C$,
any minimizer $x^*$ of $L_2$-$L_p$ problem (\ref{pp}) satisfies
$\|x^*\|_0 \ge 1$.
\end{description}
\end{theorem}
\begin{proof}
Suppose that $x^*\neq 0$ is a global  minimizer of the $L_2$-$L_p$
problem (\ref{pp}). Let  $B=A_T \in R^{m\times |T|},$ where
$T=$support$(x^*)$ and $|T|= \|x^*\|_0$ is the cardinality of the
set $T$. By Theorem 2.1 and Theorem 2.3 in \cite{CXY}, the columns
of $ B$ are linearly independent and $x^*$ must satisfy
\begin{equation}\label{sec3-1}
2B^T(Bx^*_T-b)+p\lambda (|x^*_T|^{p-2}\cdot(x^*_T))=0.
\end{equation}
This implies $Ax^*-b=Bx^*_T-b\neq 0$. Hence we have
\begin{eqnarray}\label{sec3-2}
f_p(x^*)=\|Ax^*-b\|^2+\lambda \|x^*\|^p_p> \lambda \sum_{i\in T}
|x^*_i|^p \ge  \lambda |T| \left(\frac{\lambda
p(1-p)}{2\alpha}\right)^{p/(2-p)},
\end{eqnarray}
where the last inequality is from the lower bound theory for local
minimizers of (\ref{pp}) in \cite[Theorem 2.1]{CXY}.

(1) Suppose that $\lambda \ge \beta(k).$ If $x^*$ is a nonzero
minimizer of (\ref{pp}) with $\|x^*\|_0\ge k\ge1$, then from
(\ref{sec3-2}) and the definition of $\beta(k)$ in
(\ref{alpha-beta}), we have
$$f_p(x^*)>k\lambda^{2/(2-p)}\left(\frac{p(1-p)}{2\alpha}\right)^{p/(2-p)} \ge
\|b\|^2= f_p(0).$$ This contradicts to that $x^*$ is a minimizer of
(\ref{pp}). Hence $\|x^*\|_0 < k.$

(2) Suppose $\lambda \ge \beta(1)$. If $x^*$ is a nonzero  minimizer of (\ref{pp}), then
there is $i$ such that
 $x^*_i\neq 0$ and
$$f_p(x^*)=\|Ax^*-b\|+\lambda \|x^*\|^p_p> \lambda |x^*_i|^p
\ge  \lambda \left(\frac{\lambda p(1-p)}{2\alpha}\right)^{p/(2-p)}
\ge \|b\|^2=f(0).$$ This contradicts to that $x^*$ is a minimizer of
(\ref{pp}). Hence, $x=0$ is the unique solution of (\ref{pp}).

(3) Note that $f_p(0)=\|b\|^2$ and $f_p(x_c)=\lambda \|x_c\|_p^p$
for $x_c\in C$. Therefore, if
\begin{equation}\label{nonzero}
\lambda \le \frac{\|b\|^2}{\|x_c\|_p^p} \quad \mbox{for some} \quad
x_c\in C
\end{equation}
then $f_p(0)\ge f_p(x_c)$.  Since $x_c$ is not a stationary point of
$L_2$-$L_p$ problem \cite{CXY}, there is $\tilde{x}$ near $x_c$ such
that $f_p(x_c)>f_p(\tilde{x}).$ Hence $x=0$ cannot be a global
minimizer of (\ref{pp}).
\end{proof}

\noindent
 {\bf Remark 1}\,  It was known that $x=0$ is a local minimizer of  the $L_2$-$L_p$
problem (\ref{pp})  for any value of $\lambda >0$ \cite{CXY}, and
$x=0$ is a global minimizer of (\ref{pp}) for a ``sufficiently
large'' $\lambda$ \cite{Huang}. Theorem \ref{Th3.1}, for the first
time, establishes a specific bound $ \beta(1)$, such that $x=0$ is
the unique global minimizer of (\ref{pp}) for $\lambda \ge
\beta(1)$. An important algorithmic implication of Theorem
\ref{Th3.1} is that, for given data $(A,b)$ and $p$, choosing
$\lambda\ge \beta(k)$ for a small constant $k$ does not help to
solve the original sparse least squares problem. For a small
constant $k$, say from $1$ to $3$, one might be better off to
enumerate all combinations of solutions, each with no more than $k$
nonzero entries, to find a minimizer. This can be done in a strongly
polynomial time of
the problem dimensions.\\

One may be also interested in the relation of $\lambda$ and the
support sizes of local minimizers of $L_2$-$L_p$ problem
(\ref{pp}).  We present the following result for the sparsity of certain
local minimizers of (\ref{pp}).
\begin{theorem}\label{Th2}  Let
\begin{equation}\label{gamma}
\gamma(k) =k^{p-1}\left(\frac{2\|A\|}{p}\right)^p\|b\|^{2-p}.
\end{equation}
If $\lambda \ge \gamma(k)$, then any local minimizer $x^*$ of
problem (\ref{pp}), with $f_p(x^*)\le f_p(0)=\|b\|^2$, satisfies
$\|x^*\|_0< k$ for $k\ge 2$.
\end{theorem}
\begin{proof}
Note that (\ref{sec3-1}) holds for any local minimizer of
$L_2$-$L_p$ problem (\ref{pp}). By Theorem 2.3 in \cite{CXY}, for
any local minimizer $x^*$ of $L_2$-$L_p$ problem (\ref{pp}) in the
level set $\{x:\ f_p(x)\le f_p(0)\}$, we have
\begin{eqnarray}\label{sec3-4}
f_p(x^*)=\|Ax^*-b\|^2+\lambda \|x^*\|^p_p> \lambda \sum_{i\in T}
|x^*_i|^p \ge  \lambda |T| \left(\frac{\lambda
p}{2\|A\|\|b\|}\right)^{p/(1-p)},
\end{eqnarray}
where $T=$support$(x^*).$  If $|T|=\|x^*\|_0\ge k\ge 1$, then
$$f_p(x^*)>\lambda k\left(\frac{\lambda
p}{2\|A\|\|b\|}\right)^{p/(1-p)}=\lambda^{1/(1-p)}k\left(\frac{p}{2\|A\|}\right)^{p/(1-p)}\|b\|^{p/(p-1)}
\ge \|b\|^2=f_p(0),$$ which is a contradiction.
\end{proof}

Theorem \ref{Th3.1} concerns global minimizers of $L_2$-$L_p$
problem (\ref{pp}) while Theorem \ref{Th2} concerns its local
minimizers in the level set $\{x:\ f_p(x)\le f_p(0)\}$.  Since $x=0$ is a trivial
local minimizer for problem (\ref{pp}), we believe any good method would likely find
a minimizer that at least is better than $x=0$.
Below, we use an example to illustrate the bounds presented in Theorems \ref{Th3.1}  and \ref{Th2}.

\noindent{\bf Example 2.1} Consider the following $L_2$-$L_{1/2}$
minimization problem
\begin{equation}\label{e2.1}
\mbox{Minimize} \quad
f(x):=(x_1+x_2-1)^2+\lambda(\sqrt{|x_1|}+\sqrt{|x_2|}).
\end{equation}
From $A=(1, 1)$, $b=1$ and $x_c=(1,0)$,   we easily find these data
in Theorem \ref{Th3.1} and Theorem \ref{Th2},
$$\alpha=1, \quad \|b\|=1, \quad \beta(k)=8^{1/4}k^{-3/4}, \quad
\frac{\|b\|^2}{\|x_c\|_p^p}=1, \quad \gamma(k)=32^{1/4} k^{-1/2}.$$
For $k=2$, we have $\beta(2)=1$. Using parts 1 and 3 of Theorem
\ref{Th3.1}, we can claim that any minimizer $ x^*$ of (\ref{e2.1})
with $\lambda=1$ satisfies $\|x^*\|_0=1.$ Using part 2 of  Theorem
\ref{Th3.1}, we can claim that $x=0$ is the unique minimizer of
(\ref{e2.1}) with $\lambda\ge \beta(1)=8^{1/4}.$ The lower bound
$\beta(1)$ can be improved further. In fact, we can give a number
$\beta^*\le \beta(1)$ such that $x=0$ is the unique minimizer of
(\ref{e2.1}) with $\lambda \ge \beta^*$ by using the first and
second order necessary conditions \cite{CXY} for (\ref{pp}).

For $\lambda=\frac{8}{3\sqrt{3}} < 8^{1/4}$, it is easy to see that
$(x_1,x_2)=(1/3,0)$ and $(x_1,x_2)= (0,1/3)$ are two vectors
satisfying
\begin{eqnarray*}
2x_1(x_1+x_2-1)+\frac{\lambda}{2}\sqrt{|x_1|}=0, \quad \quad
2x_2(x_1+x_2-1)+\frac{\lambda}{2}\sqrt{|x_2|}=0,
\end{eqnarray*}
 and
 $$H(x)=2\left(\begin{array}{cc}
 x_1^2 & x_1x_2\\
 x_1x_2 & x_2^2
\end{array}
\right)
 -\frac{\lambda}{4}
\left(\begin{array}{cc}
 \sqrt{|x_1|} & 0\\
 0 & \sqrt{|x_2|}
\end{array}
\right)=0.$$ However, since the third order derivative of
$g(t):=f((1/3+t)e_1)$ (or $g(t):= f((1/3+t)e_2)$) is strictly
positive on both side of $t=0$, $(x_1,x_2)=(1/3,0)$ and $(x_1,x_2)=
(0,1/3)$ are not local minimizers. Moreover, these two vectors are
the only nonzero vectors satisfying both first and second order
necessary conditions. We can claim that $x=0$ is the unique global
minimizer of (\ref{e2.1}).\\

Our theorems reinforce the findings from statistical
literature that global minimizers of the $L_2$-$L_p$
regularization problem may have many advantages over those from other
convex regularization problems, and the new results actually give precise bounds
on how to choose $\lambda$ for desirable sparsity. The remaining question:
is the $L_2$-$L_p$ regularization problem (\ref{pp}) tractable for given
$\lambda>0$ and $0\le p<1$? Or more specifically, is there an efficient or
polynomial-time algorithm to find a global minimizer of problem (\ref{pp})?
Unfortunately, we prove a strong negative result in the next section.

\section{The $L_2$-$L_p$ problem is strongly NP-hard}
As we mentioned earlier, one may attempt to draw a negative result
directly from constrained $L_p$ problem (\ref{constpp}) or
(\ref{constppepsilon}). However, it is well known that the
quadratic penalty function is not exact because its minimizer is
generally not the same as the solution of the corresponding
constrained optimization; see, e.g., \cite{Nocedal}. For example,
the all-zero vector is a local minimizer of the $L_2$-$L_p$
problem (\ref{pp}), but it may not even be feasible for the $L_p$
problem (\ref{constpp}). On the other hand, the set of all basic
feasible solutions of (\ref{constpp}) is exactly the set of its
local minimizer \cite{Ge}, but such a local minimizer of
(\ref{constpp}) may not even be a stationary point of problem
(\ref{pp}). In fact, there is no $\lambda >0$ such that $\bar{x}$,
any feasible solution of problem (\ref{constpp}), satisfies the
first order necessary condition of $L_2$-$L_p$ problem (\ref{pp}).

Another difference between (\ref{constpp}) and (\ref{pp}) is the
following: it has been shown in \cite{Ge} that any solution is a
local minimizer of (\ref{constpp}) as long as it satisfies the first
and second order necessary optimality conditions of (\ref{constpp}).
However, Example 2.1 shows that this fact is not true for
$L_2$-$L_p$ problem (\ref{pp}).

Thus, we need somewhat new proofs for the hardness result. To
facilitate the new proof, we first prove that problem (\ref{pp-q})
is NP-hard, and then extend to the strongly NP-hard result.
\begin{theorem}\label{thm:pp-np}
Minimization problem (\ref{pp-q}) is NP-hard for any given $0\le p
<1$, $q\ge 1$ and $\lambda>0$.
\end{theorem}
We first prove a useful technical lemma.

\begin{lemma}\label{lem:min-gx}
Consider the problem
\begin{eqnarray}\label{technicallemma}
\mbox{Minimize}_{z\in R} \quad g(z):=|1-z|^q+\frac{1}{2} |z|^p
\end{eqnarray}
for some given $0\le p < 1$ and $q\geq 1$. It is minimized at a
unique point (denoted by $z^*(p, q)$) on $(0,1]$. And the optimal
value $c(p ,q)$ is less than $\frac{1}{2}$.
\end{lemma}
\begin{proof}
First it is easy to see that when $p=0$, $g(z)$ has a unique
minimizer at $ z = 1$, and the optimal value is $\frac{1}{2}$. Now
we consider the case when $p\neq 0$. Note that $g(z)>g(0)=1$ for all
$z<0$, and $g(z)>g(1)=\frac{1}{2}$ for all $z>1$. Therefore the
minimum point must lie within $[0,1]$.

To optimize $g(z)$ on $[0,1]$, we check its first derivative
\begin{eqnarray}\label{derivativeg}
g'(z)=-q(1-z)^{q-1}+\frac{pz^{p-1}}{2}.
\end{eqnarray}
We have $g'(0^+)=+\infty$ and $g'(1)=\frac{p}{2}>0$. Therefore, if
function $g(z)$ has at most two stationary points in (0,1), the
first one must be a local maximum and the second one must be the
unique global minimum and the minimum value
 $c(p,q)$ must be less than $\frac{1}{2}$.

Now we check the possible stationary points of $g(z)$. Consider
solving $g'(z)=-q(1-z)^{q-1}+\frac{p z^{p-1}}{2}=0$. We get
$z^{1-p}(1-z)^{q-1}=\frac{p}{2q}$.

Define $h(z)=z^{1-p}(1-z)^{q-1}$. We have
\[h'(z)=h(z)(\frac{1-p}{z}-\frac{q-1}{1-z}).\] Note that
$\frac{1-p}{z}-\frac{q-1}{1-z}$ is decreasing in $z$ and must have a
root on $(0,1)$. Therefore, there exists a point $\bar{z}\in(0,1)$
such that $h'(z)>0$ for $z<\bar{z}$ and $h'(z)<0$ for $z>\bar{z}$.
This implies that $h(z)=\frac{p}{2q}$ can have at most two solutions
in $(0,1)$, i.e., $g(z)$ can have at most two stationary points. By
the previous discussions, the lemma holds.
\end{proof}

{\noindent\bf {Proof of Theorem \ref{thm:pp-np}.}} First we claim
that without loss of generality we only need to consider the
problem with $\lambda = \frac{1}{2}$. This is because given any
problem of form (\ref{pp-q}), we can make the following
transformation:
$$ \tilde{x} = (2\lambda)^{1/p}x \mbox{ , } \tilde{A} = (2\lambda)^{-1/p} A \mbox{ and } \tilde{b}=b$$
and scale this problem to:
\begin{eqnarray}\label{half}
\mbox{Minimize}_{\tilde{x}}\quad \|\tilde{A}\tilde{x} -
\tilde{b}\|_q^q + \frac{1}{2}\|\tilde{x}\|_p^p.
\end{eqnarray} Note that this transformation
is invertible, i.e., for any given $\lambda_0$, one can transform an
instance with $\lambda= \lambda_0$ to one with $\lambda =
\frac{1}{2}$ and vice versa. Therefore, we only need to consider the
case when $\lambda = \frac{1}{2}$.

Now we present a polynomial time reduction from the well known
NP-complete {\em partition problem}~\cite{Garey} to problem
(\ref{half}). The partition problem can be described as follows:
given a set $S$ of rational numbers $\{a_1, a_2, \ldots, a_n\}$, is
there a way to partition $S$ into two disjoint subsets $S_1$ and
$S_2$ such that the sum of the numbers in $S_1$ equals to the sum of
the numbers in $S_2$?

Given an instance of the partition problem with
$a=(a_1,a_2,\ldots,a_n)^T\in R^n$. We consider the following
minimization problem in form (\ref{half}):
\begin{equation}\label{p0r}
\begin{array}{cc}
\mbox{Minimize}_{x,y} & {\displaystyle
P(x,y)=|a^T(x-y)|^q+\sum_{1\leq j \leq n}|x_j+y_j-1|^q+
\frac{1}{2}\sum_{1\leq j \leq n}(|x_j|^p+|y_j|^p) }.
\end{array}
\end{equation}
We have
\begin{eqnarray*}
\mbox{Minimize}_{x,y} P(x,y)& \geq & \mbox{Minimize}_{x_j,y_j}
\sum_{1\leq j \leq n}|x_j+y_j-1|^q+
\frac{1}{2}\sum_{1\leq j \leq n}(|x_j|^p+|y_j|^p)\\
& = & \sum_{1\leq j \leq n} \mbox{Minimize}_{x_j,y_j} \
|x_j+y_j-1|^q + \frac{1}{2}
(|x_j|^p+|y_j|^p)\\
& = & n\cdot \mbox{Minimize}_{z}\ |1-z|^q+\frac{1}{2} |z|^p,
\end{eqnarray*}
where the last equality is from the fact that $|x_j|^p + |y_j|^p
\ge |x_j + y_j|^p$ and that we can always choose one of them to be
$0$ such that the equality holds.

By applying Lemma \ref{lem:min-gx}, we have
$$P(x,y) \ge nc(p,q).$$
Now we claim that there exists an equitable partition to the
partition problem if and only if the optimal value of (\ref{half})
equals to $nc(p,q)$. First, if $S$ can be evenly partitioned into
two sets $S_1$ and $S_2$, then we define $(x_i=z^*(p,q), y_i=0)$ if
$a_i$ belongs to $S_1$ and define $(x_i=0, y_i=z^*(p,q))$ otherwise.
These $(x_j,y_j)$ provide an optimal solution to $P(x,y)$ with
optimal value $nc(p,q)$. On the other hand, if the optimal value of
(\ref{pp-q}) is $nc(p,q)$, then in the optimal solution, for each
$i$, we must have either $(x_i=z^*(p,q), y_i=0)$ or $(x_i=0,
y_i=z^*(p,q))$. And we must also have $a^T(x-y)=0$, which implies
that there exists an equitable partition to set $S$. Thus Theorem
\ref{thm:pp-np} is proved. $\Box$

In the following, using the similar idea, we prove a stronger result:
\begin{theorem}\label{thm:pp-strongNP}
Minimization problem (\ref{pp-q}) is strongly NP-hard for any given
$0\le p <1$, $q\ge 1$ and $\lambda>0$.
\end{theorem}
\begin{proof}

We present a polynomial time reduction from the well known strongly
NP-hard 3-partition problem~\cite{Garey1,Garey}. The 3-partition
problem can be described as follows: given a multiset $S$ of $n=3m$
integers $\{a_1, a_2, \ldots, a_n\}$ with sum $mB$, can $S$ be
partitioned into $m$ subsets, such that the sum of the numbers in
each subset is equal?

We consider the following minimization problem in the form
(\ref{half}):
\begin{equation}\label{ppr}
\mbox{Minimize} \quad {\displaystyle P(x)=\sum_{j=1}^m|\sum_{i=1}^n
a_{i}x_{ij}-B|^q + \sum_{i=1}^n |\sum_{j=1}^m x_{ij}-1|^q }+
\frac{1}{2} \sum_{i=1}^n\sum_{j=1}^m |x_{ij}|^p.
\end{equation}

The remaining argument will be the same as the proof for Theorem
\ref{thm:pp-np}.
\end{proof}

Theorem \ref{thm:pp-strongNP} implies that the $L_2$-$L_p$
minimization problem is strongly NP-hard. Next we generalize the
NP-hardness result to the smoothed version of this problem in
(\ref{pp-q-e}).

\begin{theorem}\label{thm:pp-q-np} Minimization
problem (\ref{pp-q-e}) is strongly NP-hard for any give $0< p<1$,
$q\ge 1$, $\lambda>0$ and $\epsilon>0$.
\end{theorem}

\begin{proof}
We again consider the same $3$-partition problem, we claim that it
can be reduced to a minimization problem in form (\ref{pp-q-e}).
Again, it suffices to only consider the case when $\lambda =
\frac{1}{2}$ (Here we consider the hardness result for any given
$\epsilon> 0$. Note that after the scaling, $\epsilon$ may have
changed). Consider:
\begin{equation}\label{pp-q-e-eq}
\mbox{Minimize}_x \quad {\displaystyle
P_{\epsilon}(x)=\sum_{j=1}^m|\sum_{i=1}^n a_{i}x_{ij}-B|^q +
\sum_{i=1}^n |\sum_{j=1}^m x_{ij}-1|^q }+ \frac{1}{2}
\sum_{i=1}^n\sum_{j=1}^m (|x_{ij}|+\epsilon)^p.
\end{equation}
We have
\begin{eqnarray*}
\mbox{Minimize}_x P_{\epsilon}(x)& \geq & \mbox{Minimize}_x\
\sum_{i=1}^n |\sum_{j=1}^m x_{ij}-1|^q + \frac{1}{2}
\sum_{i=1}^n\sum_{j=1}^m (|x_{ij}|+\epsilon)^p\\
& = & \sum_{i=1}^n \mbox{Minimize}_x \ |\sum_{j=1}^m x_{ij}-1|^q +
\frac{1}{2} \sum_{j=1}^m (|x_{ij}|+\epsilon)^p\\ & = &
       n\cdot \mbox{Minimize}_{z}\ |1-z|^q+\frac{1}{2} (|z|+\epsilon)^p+
\frac{(m-1)}{2} {\epsilon}^p.
\end{eqnarray*}
The last equality comes from the submodularity of the function
$(x+\epsilon)^p$ and the fact that one can always choose only one
of $x_{ij}$ to be nonzero in each set such that the equality
holds. Consider function $g_{\epsilon}(z)=|1-z|^q+\frac{1}{2}
(|z|+\epsilon)^p$. Similar to Lemma \ref{lem:min-gx}, one can
prove that $g_{\epsilon}(z)$ has a unique minimizer in $[0,1]$.
Denote this minimum value by $c(p, q, \epsilon)$, we know that
$P_{\epsilon}(x)\geq n c(p, q, \epsilon)$. Then we can argue that
the 3-partition problem has a solution if and only if
$P_{\epsilon}(x)= n c(p, q, \epsilon)$. Therefore Theorem
\ref{thm:pp-q-np} holds.
\end{proof}

The above results reveal that finding a global minimizer for the
$L_q$-$L_p$ minimization problem is strongly NP-hard, or the original sparse
least squares problem is intrinsically hard, and no regularized optimization
models/methods could help much in the worst case. That is,
relaxing $L_0$ to $L_p$ for some $0<p < 1$ in the regularization
gains no significant advantage in terms of the (worst-case) computational complexity.

\section{Bounds $\beta(k)$ and $\gamma(k)$ for asymptotic properties}
Given the strong negative result for computing a global minimizer,
our hope now is to find a local minimizer of problem (\ref{pp}),
still good enough for the desired sparsity -- say no more than $k$
nonzero entries. This is indeed guaranteed by Theorem \ref{Th2} if
one chooses $\lambda\ge \gamma(k)$ of (\ref{gamma}), instead of
$\lambda\ge \beta(k)$ of (\ref{alpha-beta}). In the following, we
present a positive result in the bridge estimator model considered
by \cite{Fan01,Huang,knight}.

Consider asymptotic properties of the $L_2$-$L_p$ minimization
(\ref{pp}) where the sample size $m$ tends to infinity in the
model of \cite{Fan01,Huang,knight}. Suppose that the true
estimator $x^*$ has no more than $k$ nonzero entries. One expects
that there is a sequence of bridge estimators, i.e. solutions
$x^*_m$ of
$$\mbox{Minimize} \, \|Ax-b\|^2+\lambda_m\|x\|^p_p$$
such that  dist(support$\{x_m^*\}$, support$\{x^*\})\to 0$, as $m\to
\infty,$ with probability 1.

In applications of variable selection, the design matrix is typically standardized so that
$$ \|a_i\|^2=m \quad {\rm for} \quad i=1,\ldots, n.$$
Moreover, the smallest and largest eigenvalues $\rho_1$ and
$\rho_2$ of the covariate matrix $ \sum_m= \frac{1}{m} A^TA $
satisfy $0<c_1\le \rho_1\le \rho_2<c_2 $ for some constants $c_1$
and $c_2$, see \cite{Huang}. This assumption implies that
$\sqrt{c_1m}\le \|A\|\le \sqrt{c_2m}$. For simplicity, let us fix
$\|A\|=\sqrt{m}$ and $p=1/2$. Then we have
\[\beta(k)=k^{-3/4}(8m)^{1/4}\|b\|^{3/2}\quad\mbox{and}\quad \gamma(k)=k^{-1/2}(16m)^{1/4}\|b\|^{3/2}.\]
One can see that $\gamma(k)>\beta(k)$ for all $k\ge 1$.

If $k$ is a constant, we see that $\beta(k)$ and $\gamma(k)$ are
in the same order of $m$ and $\|b\|$. Thus, finding any local
minimizer of problem (\ref{pp}) in the objective level set
$f_p(0)$ is sufficient to guarantee desired sparsity when
$\lambda_m = \beta(k)$. That is, there is no significant
guaranteed sparsity difference between global and local minimizers
of problem (\ref{pp}). This seems also observed in computational
experiments when the true estimator is extremely sparse. Of
course, when $k$ increases as $m\to \infty$, a global minimizer of
problem (\ref{pp}) would likely become sparser than its local
minimizer, since $\beta(k)/\gamma(k)=O(k^{-1/4})$.

In general, both  $\beta(k)$ and $\gamma(k)$ meet the conditions
in the analysis of consistency and oracle efficiency of bridge estimators of \cite{Huang,knight}.
In their model, the parameter $\lambda_m$ is
required to satisfy certain conditions. For instances,
\begin{equation}\label{stand}
(\cite[Theorem \, 3]{knight}) \quad \quad \quad \lambda_m m^{-p/2}
\quad \to \quad \lambda_0 \ge 0 \quad{\rm as} \quad m \to \infty
\end{equation}
\begin{equation}\label{stand2}
(\cite[A3, (a)]{Huang}) \quad \quad \quad
 \lambda_m m^{-1/2} \quad \to \quad 0 \quad{\rm as}
\quad m \to \infty.
\end{equation}

With $\|a_i\|^2=m$ for $i=1,\ldots, n$ and $\|A\|=\sqrt{m}$ in their model, we have
$$\beta(k)m^{-p/2}=k^{p/2-1}\left(\frac{2}{p(1-p)}\right)^{p/2}\|b\|^{2-p}
\quad \to \quad \lambda_0 \ge 0 \quad{\rm as} \quad m \to \infty$$
and
$$\beta(k)m^{-1/2}=k^{p/2-1}\left(\frac{2}{p(1-p)}\right)^{p/2}\|b\|^{2-p}m^{(p-1)/2}
\quad \to \quad 0 \quad {\rm as} \quad m\to \infty.$$
For $\gamma(k)$, we have
$$\gamma(k)m^{-p/2}=
k^{p-1}\left(\frac{2}{p}\right)^p\|b\|^{2-p} \quad
\to \quad \lambda_0 \ge 0 \quad{\rm as} \quad m \to \infty$$
and
$$\gamma(k)m^{-1/2}=
k^{p-1}\left(\frac{2}{p}\right)^p\|b\|^{2-p}m^{(p-1)/2} \quad
\to \quad  0 \quad{\rm as} \quad m \to \infty.$$

Hence, both $\lambda_m=\beta(k)$ and $\lambda_m=\gamma(k)$ satisfy
(\ref{stand}) and (\ref{stand2}). Moreover, by Theorem \ref{Th3.1}
and Theorem \ref{Th2}, any minimizer of $L_2$-$L_p$ problem
(\ref{pp}) with $\lambda=\lambda_m$ is likely to have less than $k$
nonzero entries. Hence each of them could be a good choice for
consistency and oracle efficiency of bridge estimators via solving
the unconstrained $L_2$-$L_p$ minimization problem (\ref{pp}).

\end{document}